\documentclass[a4paper]{article}
\usepackage{amsmath,amssymb,amsthm,hyperref}
\usepackage{authblk}
\usepackage{enumerate}
\usepackage{todonotes}
\usepackage{tikz}

\newcommand{\pst}{\mbox{pstd}}

\newcommand{\lst}[1]{\mathcal{L}(#1)} 
\newcommand{\tl}[1]{\mathcal{C}(#1)}   
\newcommand{\lss}[2]{\mbox{lss}(#1,#2)}

\theoremstyle{plain}
\newtheorem{theorem}{Theorem}
\newtheorem{corollary}[theorem]{Corollary}
\newtheorem{proposition}[theorem]{Proposition}
\newtheorem{lemma}[theorem]{Lemma}

\theoremstyle{definition}
\newtheorem{definition}[theorem]{Definition}
\newtheorem{remark}[theorem]{Remark}
\newtheorem{example}[theorem]{Example}

\title{Some variations on Lyndon words}

\author[1]{Francesco Dolce}
\author[2]{Antonio Restivo}
\author[3]{Christophe Reutenauer}

\affil[1]{\small{
		IRIF, Universit\'e Paris Diderot (France), \\
		\url{dolce@irif.fr}
}}

\affil[2]{\small{
		Dipartimento di Matematica e Informatica, Universit\`a degli Studi di Palermo (Italy), \\ \url{antonio.restivo@unipa.it}
}}

\affil[3]{\small{
		LaCIM, Universit\'e du Qu\'ebec \`a Montr\'eal (Qu\'ebec, Canada), \\ \url{reutenauer.christophe@uqam.ca}
}}

\begin{document}

\maketitle

\begin{abstract}
In this paper we compare two finite words $u$ and $v$ by the lexicographical order of the infinite words $u^\omega$ and $v^\omega$.
Informally, we say that we compare $u$ and $v$ by the infinite order.
We show several properties of Lyndon words expressed using this infinite order.
The innovative aspect of this approach is that it allows to take into account also non trivial conditions on the prefixes of a word, instead that only on the suffixes.
In particular, we derive a result of Ufnarovskij [V. Ufnarovskij, \emph{Combinatorial and asymptotic methods in algebra}, 1995] that characterizes a Lyndon word as a word which is greater, with respect to the infinite order, than all its prefixes.
Motivated by this result, we introduce the prefix standard permutation of a Lyndon word and the corresponding (left) Cartesian tree.
We prove that the left Cartesian tree is equal to the left Lyndon tree, defined by the left standard factorization of Viennot [G. Viennot, \emph{Algèbres de Lie libres et monoïdes libres}, 1978].
This result is dual with respect to a theorem of Hohlweg and Reutenauer [C. Hohlweg and C. Reutenauer, \emph{Lyndon words, permutations and trees}, 2003].
\end{abstract}

\section{Introduction}
\label{sec:intro}

Let $A$ be a totally ordered alphabet.
A word $w$ is called a \emph{Lyndon word} if for each nontrivial factorization $w=uv$, one has $w < v$ (here $<$ is the lexicographical order).
Lyndon words were introduced in~\cite{Lyndon}.

A well-known theorem of Lyndon states that every finite word $w$ can be decomposed in a unique way as a nonincreasing product $w = \ell_1 \ell_2 \cdots \ell_n$ of Lyndon words.
This theorem, which is a combinatorial counterpart of the famous theorem of Poincaré-Birkhoff-Witt, provides an example of a factorization of the free monoid (see~\cite{L2}).
It has also many algorithmic applications and it may be computed in an efficient way. 
Indeed, Duval proposed in~\cite{D} a linear-time algorithms to compute it, while Apostolico and Crochemore proposed in~\cite{AC} a $O(\lg n)$-time parallel algorithm.

The (right) \emph{Lyndon tree} of a Lyndon word $w$ corresponds recursively to the following \emph{right standard factorization} of $w$, when no reduced to a single letter: $w$ can be written as $w=uv$ where $v$ is the longest proper nonempty suffix of $w$ which is a Lyndon word.
The word $u$ is then also a Lyndon word.
Remark that one can also define a \emph{left standard factorization} of a Lyndon word, and then a left Lyndon tree (cf.~\cite{V} and~\cite{BLRS}).

On the other hand, one can associate to a Lyndon word $w$ the \emph{Cartesian tree} corresponding to its suffix standard permutation (also commonly known as the inverse suffix array of $w$).
Hohlweg and Reutenauer have proved that the (right) Lyndon tree of a Lyndon word is equal to its Cartesian tree (see~\cite{HR}).
This connection is useful for the computation of runs in a word (see, e.g.~\cite{CR,CIKRRW,BIINTT,runs}).

In this paper we consider a new approach that uses infinite words: the relation between two finite words $u$ and $v$ is determined by the lexicographical order of the infinite words $u^\omega$ and $v^\omega$ (where $u^\omega = uuu\cdots$).
Informally, we say that we \emph{compare $u$ and $v$ using the infinite order}.

Note that one can have $u < v$ but $u^\omega > v^\omega$.
For instance, if $a < b$, one has $ab < aba$ but $(aba)^\omega < (ab)^\omega$.

This new relation between words has been used in some important results in combinatorics on words, as, for instance in the bijection of Gessel and Reutenauer (cf.~\cite{GR}), which is at the basis of some extentions of the Burrows-Wheeler transform (see~\cite{MRRS} and~\cite{K}).

We show that several properties of Lyndon words can be expressed by using the infinite order.
We prove (Corollary\ref{cor:suffix}) that a word $w$ is a Lyndon word if and only if $w^\omega < v^\omega$ for each proper suffix $v$ of $w$, i.e., $w$ is smaller, with respect to the infinite order, than all its proper suffixes.
Moreover, we show that in the classical factorization theorem by Lyndon (every word can be factorized in a unique way as a non-increasing product of Lyndon words) the product is non-increasing also with respect to the infinite order.
We also deduce (Theorems~\ref{theo:factorization2} and~\ref{theo:shortest}) new characterizations of the first and of the last factor of the factorization in Lyndon words.
 
The innovative aspect of this new approach in the study of Lyndon words is that it takes into account also conditions on the prefixes of a word, instead that only on the suffixes.
In particular, we derive (Corollary\ref{cor:U}) a result of Ufnarovskij (cf.~\cite{U}) that characterizes a Lyndon word as a word which is greater, with respect to the infinite order, than all its proper prefixes.
 
In the last section, motivated by the Ufnarovskij’s Theorem, we show that the ordering of the prefixes of a word according to the infinite order is non-trivial, and this leads to the notion \emph{prefix array} of a word. 
We then introduce the \emph{prefix standard permutation} of a word (the inverse of the prefix array) and the corresponding left Cartesian tree.
We prove, as a result which is dual with respect to that of Hohlweg and Reutenauer, that the left Cartesian tree of a Lyndon words is equal to its left Lyndon tree (Theorem~\ref{theo:equal}).
 
Some of the results of Sections~\ref{sec:infinite} and~\ref{sec:Lyndon} can be extended by considering a \emph{generalized order} relation, i.e., an order in which the comparison between two words depends on the length of their common prefix.
A word $w$ is called a \emph{generalized Lyndon word} if $w^\omega < v^\omega$ (with respect the generalized order) for each proper suffix $v$ of $w$.
Generalized Lyndon words have been introduced in~\cite{R2} and  
their theory has been further developed in~\cite{DRR18}.
A very special case of generalized order is given by the \emph{alternating order}, in which the comparison between two words depends on the \emph{parity} of the length of  their common prefix.
The generalized Lyndon words with respect to the alternating order are called \emph{Galois words}.
A bijection, similar to that of Gessel and Reutenauer, for the alternating order, has been proved in~\cite{GRR}: it leads to the definition of the \emph{Alternate Burrows-Wheeler Transform} (ABWT), that has been also studied in~\cite{GMRRS}.

\section{Definitions and notations}
\label{sec:def}

For undefined notation we refer to~\cite{L} and~\cite{L2}.
We denote by $A$ a finite alphabet, by $A^*$ the free monoid and by $A^+$ the free semigroup.
Elements of $A^*$ are called {\em words} and the identity element, denoted by $1$ is called the {\em empty word}.
We say that a word $u$ is a {\em factor} of the word $w$ if $w=xuy$ for some words $x,y$; $u$ is a {\em prefix} (resp. {\em suffix}) if $x=1$ (resp. $y=1$); it is {\em nontrivial} if $u \neq 1$ and {\em proper} if $u \neq w$. 
We say that $w=ps$ is a {\em nontrivial factorization} of $w$ whenever $p,s$ are both nonempty.
The length of a word $w = a_1 a_2 \cdots a_n$, where $a_i \in A$ for all $i$, is equal to $n$ and it is denoted by $|w|$.

A {\em period} of a word $a_1 a_2 \cdots a_n$ is a natural integer $p$ such that $a_i = a_{i+p}$ for any $i$ such that $i, i+p \in \{ 1, \ldots, n \}$; it is called {\em a nontrivial period} if  $0<p<n$.
A word having a nontrivial period is called {\em periodic}. 

We say that {\em $v$ is a fractional power of $u$} if $u = u_1 u_2$ and $v = u^k u_1$ for some nonnegative integer $k$.
In this case, one writes $v = u^r$, where $r = k + |u_1|/|u|$ is a positive rational.
Note that for $k=0$ (or $r<1$) this means that $v$ is a prefix of $u$.
Fractional powers are also known as \emph{sesquipowers} (see, e.g.,~\cite{PR}).

We say that the $v$ is {\em a strict fractional power of $u$} if $v$ is a fractional power of $u$ and, with the notations above, $k\geq 1$ or, equivalently, that $r\geq 1$.
In this case $u$ is a prefix of $v$.

\begin{example}
\label{ex:fractional}
Let $u = abcdef$.
Then $u^{2/3}=abcd$ and $u^{5/3}=abcdefabcd$.
The last one is, in particular, a strict fractional power of $u$.
\end{example}

We denote by $A^\omega$ the set of sequences over $A$, also called {\em infinite words}; such a sequence $(a_n)_{n \ge 1}$ is also written $a_1 a_2 a_3 \cdots$.
If $w$ is a (finite) word of length $n\geq 1$, $w^\omega$ denotes the infinite word having $w$ as a prefix and of period $n$.

We denote by  $A^\infty = A^* \cup A^\omega$ the set of finite and infinite words.

A {\em border} of a word $w$ of length $n$ is a word which is simultaneously a nontrivial proper prefix and suffix of $w$.
A word is called {\em unbordered} if it has no border.
It is well-known that a word has a border if and only if it is periodic.

\section{Infinite order on finite words}
\label{sec:infinite}

Given an order $<$ on the alphabet $A$, we can define the {\em lexicographical order} $<_{lex}$ (or simply $<$ when it is clear from the context) on $A^\infty$ in the following way : $u <_{lex} v$ if either $u$ is a proper prefix of $v$ (in which case $u$ must be in $A^*$) or we may write $u = pau'$, $v = pbv'$ for some words $p \in A^*$, $u',v' \in A^\infty$ and some letters $a,b \in A$ such that $a < b$.

\begin{definition}
\label{def:comparison}
Let $s, t$ be two distinct elements of $A^\omega$ such that we have a factorization $s = u_1 \cdots u_k s_0$ with $u_1, \cdots, u_k$ finite nonempty words and $s_0$ is an infinite word.
We say that {\em the comparison between $s$ and $t$ takes place within $u_k$} if $u_1 \cdots u_{k-1}$ is a prefix of $t$, but $u_1\cdots u_k$ is not.
If moreover $u_1, u_2, \ldots, u_k$ are letters we say that the comparison takes place {\em at position} $k$.
\end{definition}

Note that, when the comparison takes place within $u_k$, one may write $t = u_1 \cdots u_{k-1} u_k' t'$, for some $t' \in A^\omega$ and $u_k' \neq u_k$ such that $|u'_k| = |u_k|$.

Suppose that $u,v$ are finite nonempty words.
The following fact is well-known: one has $u^\omega = v^\omega$ if and only if $u,v$ are power of a common word, and this is true if and only if $u$ and $v$ commute (see, for instance,~\cite[Corollary 6.2.5]{L2}).

In the sequel we compare two finite words $u$ and $v$ by comparing the infinite words $u^\omega$ and $v^\omega$, with respect to the lexicographical order.
Informally, we say that we \emph{compare $u$ and $v$ with respect to the infinite order}.

Note that, given two nonempty finite words $u,v$, it is not true that $u^\omega < v^\omega \Leftrightarrow u < v$, as shown in the following example.

\begin{example}
Let us consider the two words $u=b$, $v=ba$.
Since $u$ is a prefix of $v$ we have $u < v$.
Nevertheless $u^\omega = bbb\cdots > bababa\cdots = v^\omega$.
\end{example}

In the next section we will show that this equivalence holds for Lyndon words (Theorem~\ref{theo:uvL}).

Even though we use the term ''infinite order'', when $u^\omega \neq v^\omega$ the comparison between $u^\omega$ and $v^\omega$ takes place at a position bounded by a function of the lengths of $u$ and $v$, as stated by the following lemma, which is a reformulation of the well known Fine and Wilf theorem (see, for instance,~\cite{L}).

\begin{lemma}
\label{lem:finewilf}
Let $u, v$ be nonempty words such that $u^\omega \neq v^\omega$.
Then the comparison between $u^\omega$ and $v^\omega$ takes place at a position $k \le |u| + |v| - \gcd(|u|,|v|)$.
\end{lemma}

The following example shows that the bound given in the previous lemma is tight.

\begin{example}
Let $u = abaab$ and $v = abaababa$.
Then,
$$
u^\omega  =  abaababaabab \ldots
\quad \mbox{and} \quad
v^\omega  =  abaababaabaa \ldots.
$$
One has that $v^\omega < u^\omega$ and that the comparison takes place at position $12 = |u| + |v| - 1$.
\end{example}

\begin{lemma}
\label{lem:fractional}
Let $u,v$ be nonempty words such that $u^\omega \neq v^\omega$.
Then the comparison between $u^\omega$ and $v^\omega$ takes place within the first factor $v$ of $v^\omega$ if and only if $v$ is not a fractional power of $u$.
\end{lemma}
\begin{proof}
The comparison between the two infinite words takes place within the first $v$ if and only if the two prefixes of length $|v|$ of $u^\omega$ and $v^\omega$ are different.
The conclusion follows from the fact that $v$ is a fractional power of $u$ if and only if $v$ is a prefix of $u^\omega$.
\end{proof}

\begin{lemma}
\label{lem:comparison2}
Let $s,t \in A^\omega$ be as in Definition \ref{def:comparison} (and the sentence following it).
Then $s < t$ (resp. $s > t$) implies $u_1 \cdots u_k' s' < u_1 \cdots u_k t'$ (resp. $u_1 \cdots u_k' s' > u_1\cdots u_kt'$) for any infinite words $s',t'$.
\end{lemma}

\begin{lemma}
\label{lem:p<s}
Let $u,v$ be nonempty finite words such that $u^\omega < v^\omega$ and let $x,y$ be two finite words.
Then
\begin{enumerate}
 \item[\rm (i)] if neither $u$ or $v$ is a prefix of the other, then $(ux)^\omega < (vy)^\omega$;
 \item[\rm (ii)] if $v$ is not a fractional power of $u$, then $(u^{k+1} x)^\omega < (v y)^\omega$, where $k$ is the largest integer such that $u^k$ is a prefix of $v$.
 In particular $u^\omega < (vy)^\omega$.
\end{enumerate}
\end{lemma}
\begin{proof} 
In case (i), the comparison between the two infinite words takes place within the prefix of length $\min(|u|,|v|)$.
Hence we conclude using Lemma~\ref{lem:comparison2}.

Suppose now that the hypothesis of (ii) holds.
Then we can write $u = u'au_1$ and $v = u^ku'bv_1$, with $u' \in A^*$, $a,b\in A$ such that $a\neq b$, and $u_1, v_1 \in A^*$.
Let $m = |u^k u'|$.
Since $u^\omega = u^k u' a u_1 u^\omega$ and since $u^\omega < v^\omega$, we have that $a <_{m+1} b$.
The two infinite words $u^\omega$ and $(u^{k+1}x)^\omega$ share the same prefix of length $m+1$, and the same do $v^\omega$ and $(vy)^\omega$.
Thus the comparison between between $(u^{k+1}x)^\omega$ and $(vy)^\omega$ takes place at position $m+1$.
Since $a <_{m+1} b$, we can conclude.
\end{proof}

We use several times the following observation: the opposite order $\tilde<$ of an order $<$ is also a lexicographical order.

\begin{example}
\label{ex:lexinverse}
Let $<$ be the usual lexicographical order on $\{a, b\}$, that is such that $a < b$.
Then $\tilde<$ is defined by $b~\tilde<~a$.
\end{example}

\begin{theorem}
\label{theo:usvt}
Let $u,v$ be finite nonempty words such that $u^\omega \neq v^\omega$, and $s,t$ be infinite words in $\{u,v\}^\omega$.
Then 
$$
u^\omega < v^\omega \quad \Longleftrightarrow \quad us < vt.
$$
\end{theorem}
\begin{proof}
Without loss of generality, we may assume that $u,v$ are primitive. 

Suppose that $u^\omega < v^\omega$.
Let us first consider the case when $v$ is not a fractional power of $u$.
Then we can write $u = u_1 au_2$, $v = u^k u_1 bv_2$ with letters $a\neq b$ and $k\geq 0$.
Then $u^\omega = u^k u_1a \cdots$ and $v^\omega = u^k u_1 b\cdots$.
Since $u^\omega<v^\omega$, we must have $a<b$.
Note that by hypothesis, one has either $us = u^\omega$ or $us = u^l v\cdots$, with $l\geq 1$.
In both cases $us$ begins by $u^{k+1}$.
Thus $us = u^k u_1 a \cdots$.
Moreover $vt$ begins by $v$ and hence $vt = u^k u_1 b \cdots$.
Therefore $us < vt$.

Suppose now that $v$ is a strict fractional power of $u$.
We can write $v = u^k u_1$, $u = u_1 u_2$, for some $k \geq 1$, and $u_1,u_2$ nonempty finite words (since $u^\omega \neq v^\omega$).
In particular, we have $u_1 u_2 \neq u_2 u_1$ since $u$ is primitive.
Then $u^\omega = u^k u_1 u_2 u_1 \cdots$ and $v^\omega = u^k u_1 u_1 u_2 \cdots$ (since $k \geq 1$).
From the inequality $u^\omega<v^\omega$, we deduce that $u_2 u_1 < u_1 u_2$.
We claim that $us=u^ku_1u_2u_1\cdots$: indeed, either $s=u^lv\cdots$, $1\leq l\leq \infty$, so that $u^{k+2}$ is a prefix of $us$, implying the claim; or $s=v\cdots$, so that $us$ begins by $uv=uu^ku_1=u^ku_1u_2u_1$ and the claim is true, too.
Moreover $vt = u^k u_1 u_1 u_2 \cdots$.
Therefore $us < vt$.

It remains the case where $v$ is a proper prefix of $u$. Then $u$ is either not a fractional power of $v$, or $u$ is a strict fractional power of $v$.
In this case, we have $v^\omega~\widetilde<~u^\omega$. Hence the previous arguments imply that $vt~\widetilde<~us$ and therefore $us < vt$.

Suppose now that $u^\omega < v^\omega$ does not hold. Since $u^\omega \neq v^\omega$, we have $u^\omega~\widetilde<~v^\omega$.
From the previous arguments it follows that $us~\widetilde<~vt$, hence $vt < us$ and therefore $us < vt$ does not hold.
\end{proof}

Using the previous theorem we can prove the following result (part of the it is stated, in a more general context, in~\cite{R2}, see also~\cite{DRR18}).

\begin{corollary}
	\label{cor:uuvv}
	The following conditions are equivalent for nonempty words $u,v \in A^*$:
	\begin{enumerate}
		\item[\rm (1)] $u^\omega < v^\omega$;
		\item[\rm (2)] $(uv)^\omega < v^\omega$;
		\item[\rm (3)] $u^\omega < (vu)^\omega$;
		\item[\rm (4)] $(uv)^\omega < (vu)^\omega$;
		\item[\rm (5)] $u^\omega < (uv)^\omega$;
		\item[\rm (6)] $(vu)^\omega < v^\omega$.
	\end{enumerate}
\end{corollary}
\begin{proof}
	It follows from Theorem~\ref{theo:usvt} that Condition (1) is equivalent to each of the Conditions (2), (3) and (4).

	Condition (5) is equivalent to condition (3): indeed $ u^\omega < (vu)^\omega \Leftrightarrow uu^\omega < u(vu)^\omega \Leftrightarrow u^\omega < (uv)^\omega.$
	Similarly, condition (6) is equivalent to condition (2): indeed, $ (uv)^\omega < v^\omega \Leftrightarrow  v(uv)^\omega < vv^\omega \Leftrightarrow (vu)^\omega < v^\omega.$
\end{proof}

Note that previous corollary implies a result proved by Bergman in~\cite[ Lemma 5.1]{B} (see also~\cite[p.34 and pp.101--102]{U}).

\begin{corollary}
Let $u,v$ be two finite words such that $u^\omega < v^\omega$.
Then
\begin{equation}
    u^\omega < (uv)^\omega < (vu)^\omega < v^\omega.
\end{equation}
\end{corollary}

Another immediate consequence of Corollary~\ref{cor:uuvv} is that one can define the "infinite order" without any explicit use of infinite words, as stated by the following corollary.

\begin{corollary}
\label{cor:uvvu}
    Let $u,v$ be two finite words.
    Then $u^\omega < v^\omega$ if and only if $uv < vu$.
\end{corollary}

\section{Lyndon words}
\label{sec:Lyndon}

In this section we show some fundamental properties of Lyndon words by using the infinite order, instead of the classical lexicographical order.
Moreover, the infinite order allows to introduce an innovative point of view about Lyndon words, by taking into account conditions on the prefixes of a word (instead that only on the suffixes).
In particular, we derive a result by Ufnarovskij (Corollary~\ref{cor:U}) that characterizes a Lyndon word as a word which is greater (with respect to the infinite order) than its proper prefixes.

Let us start with the classical definition of Lyndon words in terms of the usual lexicographical order (see, for instance,~\cite{D}).

\begin{definition}
\label{def:Lyndon}
A word $w$ is a Lyndon word if any one of the following three equivalent conditions holds:
\begin{enumerate}[(i)]
    \item for any nontrivial factorization $w=uv$, $u<v$,
    \item for any nontrivial factorization $w=uv$, $uv < v$,
    \item for any nontrivial factorization $w=uv$, $uv < vu$.
\end{enumerate}
\end{definition}

The following theorem provides a characterization of Lyndon words by using the infinite order.

\begin{theorem}
\label{theo:Lyndonuv}
A word $w$ is a Lyndon word if and only if, for any nontrivial factorization $w = uv$, any one of the six equivalent conditions of Corollary~\ref{cor:uuvv} holds.
\end{theorem}
\begin{proof}
From Corollary~\ref{cor:uvvu} it follows that condition (iii) of Definition~\ref{def:Lyndon} is equivalent to condition (1) of Corollary~\ref{cor:uuvv}.
\end{proof}

In next corollary (part of whose has already been proved, in a more general context, in~\cite[Proposition 2.1]{R2}, see also~\cite{DRR18}) we highlight conditions (1) and (2) of Corollary~\ref{cor:uuvv}.
Even though such conditions show some formal similarity with the conditions (i) and (ii) of Definition~\ref{def:Lyndon}, they are essentially different. 

\begin{corollary}
\label{cor:suffix}
A word $w$ is a Lyndon word if and only if one of the following condition is satisfied for any nontrivial factorization $w=uv$:
\begin{enumerate}[1.]
 \item $u^\omega < v^\omega$.
 \item $w^\omega < v^\omega$.
\end{enumerate}
\end{corollary}

The following result, due to Ufnarovskij (see also~\cite[Theorem 2, p.35]{U}) follows from Theorem~\ref{theo:Lyndonuv} and condition (5) of Corollary~\ref{cor:uuvv}.
It provides a characterization of a Lyndon word that takes into account its prefixes, instead than its suffixes.

\begin{corollary}[Ufnarovskij]
	\label{cor:U}
	A word $w$ is a Lyndon word if and only if for any nontrivial factorization $w=ps$, one has $p^\omega < w^\omega$.
\end{corollary}

\begin{example}
	The word $w = aabab$ is a Lyndon word.
	We have $a^\omega = (aa)^\omega < (aaba)^\omega < (aab)^\omega < w^\omega$.
\end{example}

The following result is classical (see, for instance~\cite{L}).

\begin{theorem}
\label{theo:factorization}
Each word in $A^*$ can be factorized in a unique way as a nonincreasing product of Lyndon words.
\end{theorem}

The term ''nonincreasing'' in the previous theorem is referred to the classical lexicographical order.
The following theorem (cf.~\cite[Theorem 8]{BMRRS}) shows that the factorization in Lyndon words is non-increasing also with respect to the infinite order.

\begin{theorem}
\label{theo:uvL}
Let $u,v$ be two Lyndon words.
Then $u < v$ if and only if $u^\omega < v^\omega$,
\end{theorem}
\begin{proof}
Suppose that $u < v$.
If $u$ is not a prefix of $v$, then obviously  $u^\omega < v^\omega$.

If $u$ is a prefix of $v$, then  $v = uy$, for some nonempty word $y$.
From $u < v$ and $v < y$ (since $v$ is a Lyndon word) we get $u^2 < uy = v$ , and thus $u^\omega < v^\omega$.

To prove the converse implication, let us suppose that $u^\omega  < v^\omega$.
By contradiction, suppose $v < u$.
By the first part of the proof we would have $v^\omega < u^\omega$, which gives us a contradiction.
\end{proof}


In the following theorem we characterize the last element of the factorization in Lyndon words by using the infinite order.

\begin{theorem}
\label{theo:factorization2}
Let $w = \ell_1 \ell_2 \cdots \ell_n$, with $\ell_i$ Lyndon words such that $\ell_1^\omega \geq \ell_2^\omega \geq \ldots \geq \ell_n^\omega$.
Then $\ell_n$ is the shortest among all nontrivial suffixes $s$ of $w$ such that $s^\omega$ is minimum.
\end{theorem}
\begin{proof}
Let $z$ be the shortest among all nontrivial suffixes of $w$ such that $s^\omega$ is minimum.
If $w = z$ then $w$ is a Lyndon word and $\ell_1 = \ell_n = z$.

Otherwise, we can write $w = uz$.
Consider the factorization of $u$ in Lyndon words: $u = \ell_1' \ell_2' \cdots \ell_k'$, with $\ell_1'^\omega \ge \ell_2'^\omega \ge \cdots \ge \ell_k'^\omega$.
By hypothesis, $(\ell_k' z)^\omega  > z^\omega$.
Then, using Corollary~\ref{cor:uuvv}, one has $\ell_k'^\omega > z^\omega$.
It follows that $\ell_1' \ell_2' \cdots \ell_k' z$ is a non-increasing factorization of $w$ in Lyndon words.
Since the factorization is unique, we have that $z = \ell_n$.
\end{proof}

\begin{example}
\label{ex:w=ababaab}
Let $w=ababaab$.
Its non-increasing factorization into Lyndon words is $w=(ab)(ab)(aab)$.
One can check that $(aab)^\omega < (abaab)^\omega < w^\omega < (ab)^\omega < (baab)^\omega < (babaab)^\omega < b^\omega$.
\end{example}


In the previous results we gave a characterization of the last element of the factorization of a word in Lyndon words.
Now, we focus on the first factor.
This result is motivated by point 1 of Corollary~\ref{cor:suffix}: the fact that a word $w$ is not a Lyndon word implies the existence of a prefix $u$
such that $u^\omega \geq v^\omega$, where $v$ is the corresponding suffix.
If one chooses the shortest prefix satisfying this property, this turns out to be the first factor in the Lyndon factorization.
In the same vein, it is motivated by Ufnarovskijj's Theorem (Corollary~\ref{cor:U} above).

\begin{theorem}
	\label{theo:shortest}
	Let $w = \ell_1 \ell_2 \cdots \ell_n$ be the nonincreasing factorization into Lyndon words of a finite nonempty word $w$.
	\begin{enumerate}
		\item[\rm 1.] The word $\ell_1$ is the shortest nontrivial prefix $p$ of $w$ such that, when writing $w=ps$, one has either $s=1$ or $p^\omega \ge s^\omega$.
		\item[\rm 2.] The word $\ell_1$ is the shortest nontrivial prefix $p$ of $w$ such that $p^\omega \geq w^\omega$.
	\end{enumerate}
\end{theorem}

In order to prove Theorem~\ref{theo:shortest} we need a preliminary result which refines Corollary~\ref{cor:uuvv} in the case of the usual lexicographical order.

Note that, for any infinite words $s,t$ such that $s < t$, with $<$ the classical order, and for any finite word $w$, one has $ws < wt$.

\begin{corollary}
	\label{cor:l_1}
	If $\ell_1, \ell_2, \ldots, \ell_n$, with $n \ge 2$, are Lyndon words such that $\ell_1^\omega \ge \ell_2^\omega \ge \cdots \geq \ell_n^\omega$, then $\ell_1^\omega \ge (\ell_{2} \cdots \ell_n)^\omega$.
\end{corollary}
\begin{proof}
	The case $n=2$, it is trivial.
	Let consider the case $n \ge 3$.
	By induction hypothesis we have $\ell_2^\omega \ge (\ell_3 \cdots \ell_n)^\omega$.
	From Corollary~\ref{cor:uuvv} it follows that $\ell_{2}^\omega \ge (\ell_{2}\cdots \ell_n)^\omega$.
	Hence, $\ell_1^\omega \ge (\ell_{2} \cdots \ell_n)^\omega$.
\end{proof}

It order to prove Theorem~\ref{theo:shortest} let us recall the well-known fact that all Lyndon words are unbordered (see, for instance,~\cite{CK}).

\begin{proof}[Proof of Theorem~\ref{theo:shortest}]
	Let us prove the first assertion.
	When $n=1$, then $w = \ell_1$ is a Lyndon word and the result is true by point 1 of Corollary~\ref{cor:suffix}.
	
	Suppose now that $n \geq 2$.
	Then, by Corollary~\ref{cor:l_1}, we have $\ell_1^\omega \ge (\ell_{2} \cdots \ell_n)^\omega$.
	Let $p$ be a nontrivial prefix of $w$ shorter then $\ell_1$.
	Thus, we have a nontrivial factorization $\ell_1 = pq$ for some $q \ne 1$.
	By Corollary~\ref{cor:suffix}, we know that $p^\omega < q^\omega$.
	Since $\ell_1$ is unbordered, $q$ cannot be a fractional power of $p$.
	Thus, by point (ii) of Lemma~\ref{lem:p<s}, one has $p^\omega < (q \ell_2 \cdots \ell_n)^\omega$, which prove the first part of the theorem.
	
	The second assertion just follows from the first one.
	Indeed, using Corollary~\ref{cor:uuvv}, we have that if $s \ne 1$, then $p^\omega \ge s^\omega$ is equivalent to $p^\omega \ge (ps)^\omega$.
\end{proof}

\begin{example}
	\label{ex:somegalois}
	Let $w = ababaab$.
	As seen in Example~\ref{ex:w=ababaab}, its nonincreasing factorization into Lyndon words is $w = (ab)(ab)(aab)$.
	One can check that $(ab)^\omega > w^\omega > (abaab)^\omega$ while $a^\omega < w^\omega < (babaab)^\omega$.
\end{example}

\section{Left Lyndon tree and prefix standardization}
\label{sec:tree}

In~\cite{HR}, the authors associate with each Lyndon word $w$ an increasing tree, based on the suffixes of $w$; they show that the completion of this tree, with leaves appropriately labeled by the letters of $w$, is equal to the tree obtained by iterating the right standard factorization of $w$ (equivalently the Lie bracketing associated with $w$).

In this section we give a similar construction and result, based on the prefixes of $w$ instead that on the suffixes.
This construction is motivated by Ufnarovskij's Theorem (Corollary ~\ref{cor:U}).

\subsection{Left Lyndon tree}

In~\cite{V} (cf. also~\cite{BLRS}) Viennot introduced the notion of left standard factorization of a Lyndon word.
Let us consider a Lyndon word $w$ having length at least 2.
The {\em left standard factorization} of $w$ is the factorization $w = uv$, where $u$ is the longest nonempty proper prefix of $w$ which is a Lyndon word.

The following is a well-known result (see~\cite[p. 14]{V}).

\begin{proposition}
\label{pro:lyndonleft}
Let $w=uv$ be a Lyndon word with its left standard factorization.
Then both $u$ and $v$ are Lyndon words and $u<v$.
Moreover, either $v$ is a letter, or given its left standard factorization $v= v_1 v_2$ one has $v_1 \le u$.
\end{proposition}

\begin{corollary}
\label{cor:uvv1}
Let $u, v, v_1$ and $v_2$ as in Proposition~\ref{pro:lyndonleft}.
Then $v_1$ is a prefix of $u$.
\end{corollary}
\begin{proof}
By Proposition~\ref{pro:lyndonleft} we have $v_1 \le u < v_1 v_2$.
By a classical property of lexicographical order, this implies that $u = v_1 m$, for a certain word $m$ such that $m < v_2$.
Thus $v_1$ is a prefix of $u$.
\end{proof}

\begin{example}
\label{ex:leftstd}
Let us consider the Lyndon word $w = aabaacab$ on the alphabet $\{ a,b,c\}$ with $a < b < c$.
Its left standard factorization is $w = (aabaac)(ab)$.
The suffix $v=ab$ is a Lyndon word with left standard factorization $(a)(b)$ and one has $a \le aabaac$.
\end{example}

The {\em free magma} $M(A)$ over $A$ is the set of {\em complete trees} over $A$ defined recursively as follows:
\begin{itemize}
	\item[-] each letter is a tree;
	\item[-] if $\mathfrak{t}_1, \mathfrak{t}_2$ are trees, then $(\mathfrak{t}_1, \mathfrak{t}_2)$ is a tree.
\end{itemize}

We will use the classical notions of {\em root}, {\em internal node} and {\em leaf} for a tree.

There is a canonical surjective mapping $\varphi$ from $M(A)$ onto $A^+$ defined as follows:
given a tree $\mathfrak{t}$, its image $\varphi(\mathfrak{t})$, called its {\em foliage}, is defined recursively by:
\begin{itemize}
	\item $\varphi(a) = a$ for any $a$ in $A$;
	\item $\varphi((\mathfrak{t}_1, \mathfrak{t}_2))=\varphi(\mathfrak{t}_1) \varphi(\mathfrak{t}_2)$ for any two trees $\mathfrak{t}_1, \mathfrak{t}_2$.
\end{itemize}

In other words, $\varphi(\mathfrak{t})$ is obtained by vertically projecting the leaves of $\mathfrak{t}$ onto a horizontal line.

\begin{example}
The foliage of the tree in Figure~\ref{fig:foliage} is the word $aabaacab$.
\end{example}

\begin{figure}
	\centering
	\begin{tikzpicture}[
	scale=0.77, transform shape,
	every node/.style={anchor=south},
	level distance=15mm,
	level 1/.style={sibling distance=80mm, font=\footnotesize},
	level 2/.style={sibling distance=40mm, font=\footnotesize},
	level 3/.style={sibling distance=20mm,font=\footnotesize},
	level 4/.style={sibling distance=12mm,font=\footnotesize}
	]
	\node[font=\footnotesize] (ac) {$$}
	child {
		node {$x_1$}
		child {
			node{$$}
			child {
				node {$a$}
			}
			child {
				node{}
				child {
					node{$a$}
				}
				child {
					node{$b$}
				}
			}
		}
		child {
			node {$x_2$}
			child {
				node{$a$}
			}
			child {
				node{$x_3$}
				child {
					node{$a$}
				}
				child {
					node{$c$}
				}
			}
		}
	}
	child {
		node {$$}
		child {
			node{$a$}
		}
		child {
			node {$b$}
		}
	}
	;
	\end{tikzpicture}
	\caption{The left Lyndon tree of $w=aabaacab$.}
	\label{fig:foliage}
\end{figure}
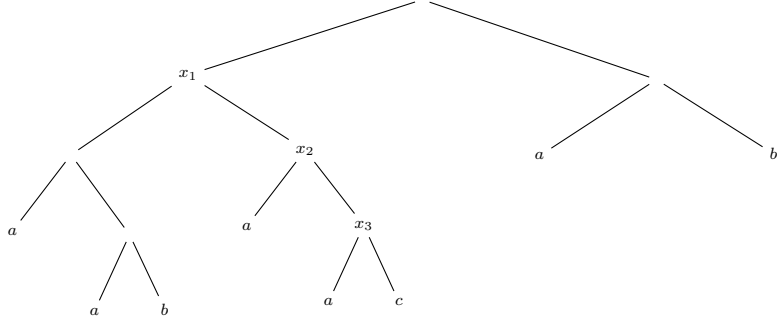

To each Lyndon word in $A^+$ we can associate a complete tree $\lst{w}$ in $M(A)$, called the
\emph{left Lyndon tree}
of $w$, defined recursively as follows:
\begin{itemize}
	\item $\lst{a} = a$ for each letter $a \in A$;
	\item $\lst{w} =(\lst{u}, \lst{v})$ for each Lyndon word $w$ of length at least 2 with left standard factorization $w=uv$.
\end{itemize}

It is clear by the definition that $\varphi(\lst{w})=w$.

\begin{example}
The left Lyndon tree associated to the Lyndon word $w=aabaacab$ is shown in Figure~\ref{fig:foliage} (disregarding the labels of the internal nodes).
\end{example}

\subsection{Prefix standardization}
\label{sec:standardization}

Consider an alphabet $A$ and a total order $<$ on $A$.
We define an order $\prec$ on the free monoid as follows: $u \prec v$ if either
\begin{itemize}
    \item $u^\omega < v^\omega$, or
    \item $u^\omega = v^\omega$ (which means that $u,v$ are power of the same word) and if $u$ is longer than $v$.
\end{itemize}

\begin{example}
	Let us consider the above order on $\{ a,b \}^*$ induced by $a<b$ .
	One has  $aa \prec a \prec ab \prec ba \prec b$.
\end{example}

Let $w$ be a word of length $n+1$ on a totally ordered alphabet.
Let us consider the sequence $p_1, p_2, \ldots, p_n, p_{n+1}=w$ of its nonempty prefixes, in increasing length.

Motivated by Ufnarovskij's Theorem (Corollary~\ref{cor:U}), we call {\em prefix standard permutation} of $w$ the unique permutation $\sigma\in \mathfrak{S}_{n+1}$ such that
$$ p_{\sigma^{-1}(1)} \prec p_{\sigma^{-1}(2)} \prec \ldots \prec p_{\sigma^{-1}(n)} \prec p_{\sigma^{-1}(n+1)}. $$
In other words, we number each letter in $w$ by $1, 2, \ldots, n, n+1$ as follows: 1 for the letter where the $\prec$-smallest prefix ends, 2 for the second smallest one in the order $\prec$, and so on.
Then $\sigma$ is the represented by the word $w$ translated in this new alphabet.
We denote such a permutation by $\pst(w)$.

\begin{example}
	The prefix standard permutation of the word $w=aabaacab$ is the permutation $\pst(w) = 21543768$.
	Indeed on can check that $aa \prec a \prec aabaa \prec aaba \prec aab \prec aabaaca \prec aabaac \prec w$.
\end{example}

\begin{remark}
In~\cite{HR} the authors define the "suffix standard permutation" of a word.
In literature, this is also commonly known as "inverse suffix array".
Therefore, our prefix standard permutation could also be called "inverse prefix array".
\end{remark}

An {\em injective word} is a word without repetition of letters.
Each permutation in $\mathfrak{S}_n$ can be seen as a complete injective word, that is an injective word in the ordered alphabet $\{1 < 2 < \ldots < n \}$, having all $n$ letters as factors.

With each injective word $\alpha$ on a totally ordered alphabet, and in particular with each permutation $\alpha \in \mathfrak{S}_n$, we can associate bijectively a binary noncomplete labeled {\em decreasing tree} (i.e., such that each node is smaller than its father) defined recursively as follows:
\begin{itemize}
    \item the root is $n$, the maximum letter in $\alpha$;
	\item its left and right subtrees (which may be empty) are the trees associated to $u$ and $v$ respectively, where $u,v$ are injective words in $\{ 1,2, \ldots, n \}$ such that $\alpha = unv$.
\end{itemize}

The inverse bijection is obtained by vertically projecting the labels on the horizontal line.

\begin{example}
	The decreasing tree associated to $\alpha = 2154376$ is the tree shown in Figure~\ref{fig:decreasing} restricted to its internal nodes.
\end{example}

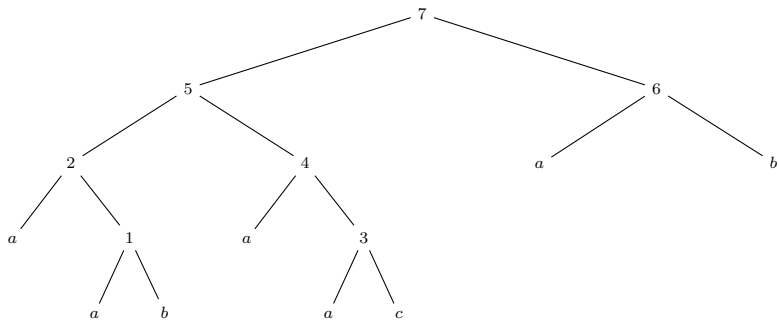
\begin{figure}
	\centering
	\begin{tikzpicture}[
	scale=0.77, transform shape,
	every node/.style={anchor=south},
	level distance=15mm,
	level 1/.style={sibling distance=80mm, font=\footnotesize},
	level 2/.style={sibling distance=40mm, font=\footnotesize},
	level 3/.style={sibling distance=20mm,font=\footnotesize},
	level 4/.style={sibling distance=12mm,font=\footnotesize}
	]
	\node[font=\footnotesize] (ac) {$7$}
	child {
		node {$5$}
		child {
			node{$2$}
			child {
				node {$a$}
			}
			child {
				node{1}
				child {
					node{$a$}
				}
				child {
					node{$b$}
				}
			}
		}
		child {
			node {$4$}
			child {
				node{$a$}
			}
			child {
				node{$3$}
				child {
					node{$a$}
				}
				child {
					node{$c$}
				}
			}
		}
	}
	child {
		node {$6$}
		child {
			node{$a$}
		}
		child {
			node {$b$}
		}
	}
	;
	\end{tikzpicture}
	\caption{The tree $\tl{aabaacab}$.}
	\label{fig:decreasing}
\end{figure}

Let us now associate a tree with a Lyndon word, called the \emph{left Cartesian tree} of $w$, and denoted $\tl{w}$, as follows (see also \cite{CR} and the references there).
To each letter $a \in A$ we define $\tl{a} = a$.
Otherwise, let $\sigma = \pst(w) \in \mathfrak{S}_{n+1}$, with $|w|=n+1$, and let $\alpha \in \mathfrak{S}_n$ be the permutation obtained by removing the last digit in $\sigma$ (which is $n+1$, because of Corollary~\ref{cor:U}).
Using the construction before we obtain a noncomplete binary decreasing tree $\mathfrak{t}^*$ having $n$ nodes.
We define $\tl{w}$ as the complete binary tree having $\mathfrak{t}^*$ as tree of internal nodes and the letters of $w$ as leaves in order from left to right, i.e., such that $\varphi(\mathfrak{t}) = w$.

\begin{example}
The tree $\tl{w}$ for $w = aabaacab$ is shown in Figure~\ref{fig:decreasing}.
\end{example}

For our purpose, we give an alternative construction of $\tl{w}$. Consider the sequence of proper prefixes $(p_1, p_2, \ldots, p_n)$ of $w$, viewed as a word of length $n$ on the alphabet $A^*$, totally ordered by $\prec$.
Since it is an injective word, we can associate with it a decreasing tree as before; call it $t^*(w)$.
Also as before, we can complete $t^*(w)$ to a tree $t(w)$, having the letters of $w$ as leaves in ordrer from left to right, i.e., such that $\varphi(t(w)) = w$.
The completed tree $t(w)$, disregarding the labels of the internal nodes, coincides with $\tl{w}$.
Indeed, the word $\alpha$ coincides with the previous injective word, up to the unique increasing order isomorphism from $\{ 1, \ldots, n \}$ into the set of proper prefixes of $w$, sending $i$ to the $i$-th prefix for the order $\prec$.

\subsection{Equivalence of the trees}

From the constructions seen before it is clear that $\varphi(\lst{w}) = \varphi(\tl{w})$ for every Lyndon word $w$.
We actually have a stronger result.

\begin{theorem}
\label{theo:equal}
Let $w$ be a Lyndon word.
The trees $\lst{w}$ and $\tl{w}$ are equal.
\end{theorem}

To prove the theorem we need some intermediate result.
Let $x$ be an internal node of some planar binary complete tree $\mathfrak{t} = ( \mathfrak{t}_1 ,\mathfrak{t}_2)$.
We call {\em left subtrees sequence} with respect to the tree $\mathfrak{t}$ and the node $x$, denoted $\lss{\mathfrak{t}}{x}$, the sequence of subtrees of $\mathfrak{t}$ hanging at the left of the path from the root to $x$, that is the sequence of subtrees of $\mathfrak{t}$ recursively defined as follows:
\begin{itemize}
    \item if $x$ is the root, then $\lss{\mathfrak{t}}{x} = (\mathfrak{t}_1)$;
    \item if $x$ is in $\mathfrak{t}_1$, then $\lss{\mathfrak{t}}{x} = \lss{\mathfrak{t}_1}{x}$;
    \item if $x$ is in $\mathfrak{t}_2$, then $\lss{\mathfrak{t}}{x} = (\mathfrak{t}_1, \lss{\mathfrak{t}_2}{x})$.
\end{itemize}

\begin{lemma}
\label{lem:lst1}
Let $w$ be a Lyndon word and let $x$ be an internal node of $\lst{w}$.
Let $\lss{\mathfrak{t}}{x} = (\mathfrak{t}_1, \ldots, \mathfrak{t}_n)$ and let $\ell_i$ be the foliage of the tree $\mathfrak{t}_i$ for each $i$.
Then all $\ell_i$ are Lyndon words.
Moreover, $\ell_{i+1}$ is a prefix of $\ell_i$ for each $1 \le i \le n-1$.
\end{lemma}
\begin{proof}
The fact that each $\ell_i$ is a Lyndon word follows from the more general fact that the foliage of each subtree of $\lst{w}$ is a Lyndon word, as follows recursively from the construction of $\lst{w}$.

We prove the other assertion by induction on the size of the tree.
If $x$ is the root, then $n=1$ and there is nothing to prove.
	
Let us suppose that $x$ is not the root of the tree and let us write $\lst{w} = (\mathfrak{t}',\mathfrak{t}'')$.
	
If $x$ is an internal node of $\mathfrak{t}'$, then $\lss{\mathfrak{t}}{x} = \lss{\mathfrak{t}'}{x}$, and we can conclude by induction.

Suppose now that $x$ is an internal node of $\mathfrak{t}''$, and let $\mathfrak{t}'' = (\mathfrak{s}_1, \mathfrak{s}_2)$.
Thus $\lss{\mathfrak{t}}{x} = (\mathfrak{t}_1, \mathfrak{t}_2, \ldots, \mathfrak{t_n})$, with $\mathfrak{t}_1 = \mathfrak{t}'$ and $(\mathfrak{t}_2, \ldots \mathfrak{t}_n) = \lss{\mathfrak{t}''}{x}$ (note that $n \ge 2)$.
By induction it is enough to show that $\ell_2$ is a prefix of $\ell_1$.
Let $v_1, v_2$ be respectively the foliages of $\mathfrak{s}_1, \mathfrak{s}_2$.
Then by the property of the left standard factorization and by the construction of the tree $\lst{w}$, we have that $v_1$ is a prefix of $\ell_1$.
Now either $x$ is in $\mathfrak{s}_2$ and $\mathfrak{t}_2 = \mathfrak{s}_1$ and $\ell_2 = v_1$
or $x$ is in $\mathfrak{s}_1$ and $\ell_2$ is a proper prefix of $v_1$
In both cases $\ell_2$ is a prefix of $\ell_1$.
\end{proof}

\begin{example}
Let us consider the tree $\mathfrak{t}$ in Figure~\ref{fig:foliage} and its internal node $x_3$.
The left subtrees sequence $\lss{\mathfrak{t}}{x_3}$ is equal to $(\mathfrak{t}_1, \mathfrak{t}_2, \mathfrak{t}_3)$, where $\mathfrak{t}_i$ is the subtree of $\mathfrak{t}$ having root $x_i$.
The foliages of the three subtrees are respectively the words $\ell_1 = aab$, $\ell_2 = a$ and $\ell_3 = a$.
Each of them is a prefix of the previous one.
\end{example}

\begin{lemma}
\label{lem:lst2}
Let $\ell_1, \ldots, \ell_n$, be Lyndon words such that $\ell_{i+1}$ is a prefix of $\ell_i$ for each $i = 1, \ldots, n-1$.
Let $\ell_n = \ell_n' \ell_n''$ be the left standard factorization of $\ell_n$.
Then
\begin{enumerate}[(i)]
    \item $(\ell_1 \cdots \ell_{n-1} \ell_n')^\omega ~ < ~ (\ell_1 \cdots \ell_n)^\omega$;
	\item moreover, if $n \ge 2$, one has $(\ell_1 \cdots \ell_n)^\omega ~ \leq ~ (\ell_1 \cdots \ell_{n-1})^\omega$.
\end{enumerate}
\end{lemma}
\begin{proof}
If $n=1$, we have $\ell_1'^\omega < \ell_1''^\omega$ by Corollary~\ref{cor:suffix} and $\ell_1'^\omega < (\ell_1' \ell_1'')^\omega$ by Corollary~\ref{cor:uuvv}.
Hence $\ell_1'^\omega < \ell_1^\omega$.
	
Suppose now that $n \ge 2$.
Since both $(\ell_1 \cdots \ell_{n-1} \ell_n')^\omega$ and $(\ell_1 \cdots \ell_n)^\omega$ start with the word $\ell_1 \cdots \ell_{n-1} \ell_n'$, we have
$$
(\ell_1 \cdots \ell_{n-1} \ell_n')^\omega ~ < ~ (\ell_1 \cdots \ell_n)^\omega
\quad \Longleftrightarrow \quad
(\ell_1 \cdots \ell_{n-1} \ell_n')^\omega ~ < ~ \ell_n'' (\ell_1 \cdots \ell_n)^\omega.
$$
Since $\ell_n$ is a prefix of $\ell_1$, the last inequality is of the form $\ell_n s_0 < \ell_n'' t_0$, with $s_0, t_0 \in A^\omega$, and this is true since $\ell_n < \ell_n''$ by Corollary~\ref{cor:suffix} and because $\ell_n$ is not a prefix of $\ell_n''$.
	
Let us now prove point (ii).
If $\ell_1 = \ell_2 = \cdots = \ell_n$, we trivially have $(\ell_1 \cdots \ell_n)^\omega = (\ell_1 \cdots \ell_{n-1})^\omega$.
Thus, let us suppose that $\ell_1 \neq \ell_n$.
Since both terms of the inequality start with $\ell_1 \cdots \ell_{n-1}$, we have
$$
(\ell_1 \cdots \ell_n)^\omega ~ \le ~ (\ell_1 \cdots \ell_{n-1})^\omega
\quad \Longleftrightarrow \quad
\ell_n (\ell_1 \cdots \ell_n)^\omega ~ \le ~ (\ell_1 \cdots \ell_{n-1})^\omega.
$$
Since $\ell_n$ is a prefix of $\ell_1$, and $\ell_n \ne \ell_1$ we can write $\ell_1 = \ell_n v$ for a nonempty finite word $v$.
Thus, the last inequality is of the form $\ell_n \ell_1 s_0 \le \ell_n v t_0$ with $s_0, t_0 \in A^\omega$, and this is equivalent to $\ell_1 s_0 \le v t_0$, which is true because of Corollary~\ref{cor:suffix}.
\end{proof}

Before proving the main result let us introduce the following notation.
Let $\mathfrak{t} = (\mathfrak{t}_1, \mathfrak{t}_2)$ be a complete binary labeled tree.
To each internal node $x$ of $\mathfrak{t}$ we associate the word $g_\mathfrak{t}(x)$, called the {\em left foliage} of $x$ in $\mathfrak{t}$, as follows:
\begin{itemize}
	\item if $x$ is the root of $\mathfrak{t}$, then $g_\mathfrak{t}(x) = \varphi(\mathfrak{t}_1)$;
	\item if $x$ is in $\mathfrak{t}_1$, then $g_\mathfrak{t}(x) = g_{\mathfrak{t}_1}(x)$;
	\item if $x$ is in $\mathfrak{t}_2$, then $g_\mathfrak{t}(x) = \varphi(\mathfrak{t}_1) g_{\mathfrak{t}_2}(x)$.
\end{itemize}

For further use, we note that the length of $g_\mathfrak{t}(x)$ is equal to the number of leaves located at the left of $x$ in $\mathfrak{t}$.

Also, we use the following result.

\begin{lemma}
\label{lem:lssg}
    Let $\mathfrak{t}$ be a complete binary labeled tree and $x$ an internal note of $\mathfrak{t}$.
    Then $g_{\mathfrak{t}}(x)=\ell_1\cdots \ell_n$, where $\ell_1,\ldots, \ell_n$ are the foliages of the trees hanging at the left on the path from the root to $x$.
\end{lemma}
\begin{proof}
    Let $\lss{\mathfrak{t}}{x} = ( \mathfrak{t}_1, \ldots, \mathfrak{t}_n)$.
    We have to show that
    $g_\mathfrak{t}(x) = \varphi(\mathfrak{t}_1) \cdots \varphi(\mathfrak{t}_n)$.
    This follows easily from the recursive definition of both $lls$ and $g$.
\end{proof}

\begin{example}
The label shown in Figure~\ref{fig:variant} is obtained by relabeling the internal nodes of the tree $\lst{w}$ in Figure~\ref{fig:foliage}, using the left foliage function.
\end{example}

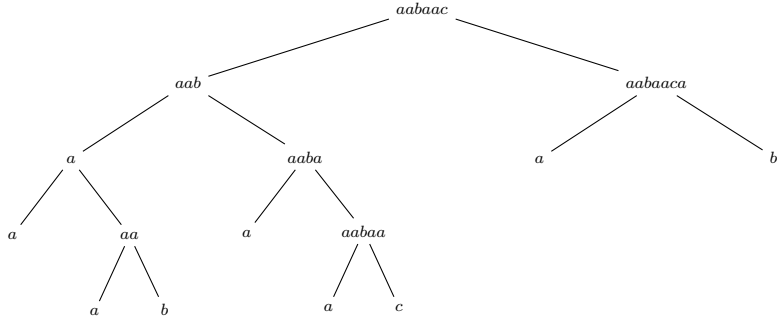
\begin{figure}
	\centering
	\begin{tikzpicture}[
	scale=0.77, transform shape,
	every node/.style={anchor=south},
	level distance=15mm,
	level 1/.style={sibling distance=80mm, font=\footnotesize},
	level 2/.style={sibling distance=40mm, font=\footnotesize},
	level 3/.style={sibling distance=20mm,font=\footnotesize},
	level 4/.style={sibling distance=12mm,font=\footnotesize}
	]
	\node[font=\footnotesize] (ac) {$aabaac$}
	child {
		node {$aab$}
		child {
			node{$a$}
			child {
				node {$a$}
			}
			child {node{$aa$}
				child {node{$a$}}
				child {node{$b$}}
			}
		}
		child {
			node {$aaba$}
			child {
				node{$a$}
			}
			child {
				node{$aabaa$}
				child {
					node{$a$}
				}
				child {
					node{$c$}
				}
			}
		}
	}
	child {
		node {$aabaaca$}
		child {
			node{$a$}
		}
		child {
			node {$b$}
		}
	};
	\end{tikzpicture}
	\caption{Variant of the left Lyndon tree of $w = aabaacab$ labeling the internal nodes with their left foliage.}
	\label{fig:variant}
\end{figure}

\begin{proof}[Proof of Theorem~\ref{theo:equal}]
Let us consider the two trees $\lst{w}$ and $\tl{w}$.
Note first that they have the same foliage.
So, in order to prove the equality, it is enough to show that the two trees obtained from $\lst{w}$ and $\tl{w}$ by removing the leaves (that is, considering only the internal nodes) are equal.

We actually show that by labeling the internal nodes of $\lst{w}$ using the function $g_\mathfrak{t}$, we obtain the same tree as $\tl{w}$.
For this it is enough to show that the labeling $g_\mathfrak{t}$ is decreasing, and that its projection is exactly the sequence $(p_1,\ldots,p_n)$ of nonempty prefixes of $w$.
Indeed, the decreasing tree associated with an injective word on a totally alphabet is unique.

Let us consider two internal nodes $x$ and $y$ of $\mathfrak{t}$, such that $y$ is a child of $x$.
We show that $g_\mathfrak{t}(y) \prec g_\mathfrak{t}(x)$, i.e., that either
\begin{itemize}
    \item $g_\mathfrak{t}(y)^\omega ~ < ~ g_\mathfrak{t}(x)^\omega$, or
    \item $g_\mathfrak{t}(y)^\omega = g_\mathfrak{t}(x)^\omega$ and 
	$|g_\mathfrak{t}(y)| ~ > ~ |g_\mathfrak{t}(x)|$.
\end{itemize}
Suppose first that $y$ is a right child of $x$ and denote by $h_1, \ldots, h_n$ the foliages of the subtrees in the sequence $\lss{\mathfrak{t}}{y}$.
We have $n \geq 2$, and by Lemma~\ref{lem:lssg} $g_\mathfrak{t}(x) = h_1 \cdots h_{n-1}$, and $g_\mathfrak{t}(y) = h_1 \cdots h_{n}$.
Thus $g_\mathfrak{t}(y) \prec g_\mathfrak{t}(x)$ follows from Lemmata~\ref{lem:lst1} and~\ref{lem:lst2} and from the fact that $|g_\mathfrak{t}(y)| > |g_\mathfrak{t}(x)|$.

Suppose now that $y$ is a left child of $x$.
Let $\lss{\mathfrak{t}}{x} = (\mathfrak{t}_1, \ldots, \mathfrak{t}_n)$, with $\mathfrak{t}_n = (\mathfrak{t}_n', \mathfrak{t}_n'')$.
Then $\lss{\mathfrak{t}}{y} = ( \mathfrak{t}_1, \ldots, \mathfrak{t}_{n-1}, \mathfrak{t}_n')$.
Let $h_i$ be the foliage of $\mathfrak{t}_i$ for each $1 \le i \le n$, and $h_n'$ be the foliage of $\mathfrak{t}_n'$.
Thus $g_\mathfrak{t}(x) = h_1 \cdots h_n$ and $g_\mathfrak{t}(y) = h_1 \cdots h_{n-1} h'_n$.
Thus $g_\mathfrak{t}(x) \prec g_\mathfrak{t}(y)$ by the Lemmata~\ref{lem:lst1} and~\ref{lem:lst2}.
	
It remains to show that the projection is exactly $(p_1,\ldots,p_n)$.
This follows from the fact the length of $g_\mathfrak{t}(x)$ is equal to the number of leaves located at the left of a given node $x$; hence $g_\mathfrak{t}(x)$ is the prefix of $w$ of length this number.
We conclude because the lengths of the successive projections of the internal nodes increase by 1 from left to right.
	\end{proof}

\begin{example}
Let $w = aabaacab$.
The tree $\lst{w}$ with each internal node $x$ labeled by $g_{\lst{w}}(x)$ (as described in the proof of Theorem~\ref{theo:equal}) is shown in Figure~\ref{fig:variant}.
\end{example}

\bibliography{biblio}

\end{document}